\documentclass[final]{dmtcs-episciences}

\usepackage{amsmath,amssymb}
\usepackage[T1]{fontenc}
\usepackage[utf8]{inputenc}
\usepackage[english]{babel}
\usepackage{color}
\usepackage{latexsym}
\usepackage{enumerate}
\usepackage{tabularx}
\usepackage{url}
\usepackage{hyperref}
\usepackage{tikz}
\usepackage{graphicx, caption}
\usepackage{subcaption}
\usepackage{marvosym}

\DeclareMathOperator\RR{\mathbb{R}}

\DeclareMathOperator\HH{\mathcal{H}}

\newenvironment{proof}{\par \noindent \textbf{Proof.} }{\hfill$\Box$\medskip}

\newtheorem{theorem}{Theorem}
\newtheorem{corollary}[theorem]{Corollary}

\newtheorem{problem}[theorem]{Problem}

\newtheorem{lemma}[theorem]{Lemma}

\newtheorem{remark}[theorem]{Remark}
\newtheorem{claim}[theorem]{Claim}

\graphicspath{{fig/}}

\title[CBU graphs]{Contact graphs of boxes with unidirectional contacts\thanks{This research is partially supported by ANR project GATO (ANR-16-CE40-0009), and ANR project GRALMECO (ANR-21-CE48-0004-01).}}
\author[D. Gon\c{c}alves \and V. Limouzy \and P. Ochem]{Daniel Gon\c{c}alves\affiliationmark{1}
\and Vincent Limouzy\affiliationmark{2}
\and Pascal Ochem\affiliationmark{1}}
\affiliation{
LIRMM, Universit\'e de Montpellier, CNRS, Montpellier, France\\
Universit\'e Clermont Auvergne, Clermont Auvergne INP, CNRS, Mines Saint-Etienne, \textsc{Limos}, F-63000 Clermont-Ferrand, France}


\begin{document}
\publicationdata{vol. 25:3 special issue ICGT'22}{2023}{5}{10.46298/dmtcs.10805}{2023-01-13; 2023-01-13; 2023-06-21}{2023-10-19}
\maketitle

\begin{abstract}
This paper is devoted to the study of particular geometrically defined intersection classes of graphs. Those were previously studied by Magnant and Martin, who proved that these graphs have arbitrary large chromatic number, while being triangle-free. We give several structural properties of these graphs, and we raise several questions.
\end{abstract} 

\section{Introduction}

A lot of graph classes studied in the literature are defined by a geometric model, 
where vertices are represented by geometric objects (e.g. intervals on a line, disks in the plane, 
chords inscribed in a circle...)  and the adjacency of two vertices is determined according 
to the relation between the corresponding objects.  A large amount of graph classes consider 
the intersection relation (e.g. interval graphs, disk graphs or circle graphs). However 
some other relations might be considered such as the containment, the overlap or also the contact
between objects. 
Recently several groups of authors started to study graph classes defined by contact models, 
as for example Contact of Paths in a grid (CPG) \cite{DenizGMR18}, Contact of $L$ shapes in $\mathbb{R}^2$
or even contact of triangles in the plane \cite{FraysseixOR94}. 
In this note we consider a new class defined by a contact model. More precisely we consider the
class of graphs defined  by contact of axis parallel boxes in $\mathbb{R}^d$ where the contact occurs on $(d-1)$-dimensional 
object in only one direction (CBU) .

When considering  graph defined by axis-parallel boxes in $\mathbb{R}^d$ and the adjacency 
relation is given by the intersection it corresponds to the important notion of boxicity 
introduced by Roberts \cite{Roberts69}, when the adjacency relation is given by the 
containment relation it correspond to comparability graphs and it is connected to the 
poset dimension introduced by Dushnik \& Miller \cite{DushnikM41}

The motivation for this class of graphs originate from an article of Magnant and Martin \cite{MagnantM11} 
where a wireless channel assignment is considered. The problem considers rectangular rooms in a building 
and asks to find a channel assignment for each room. In order to avoid interferences 
rooms sharing the same wall, floor or ceiling need to use different channels. The question was to determine 
whether a constant number of channels would suffice to answer this problem. The first negative answer 
was provided by Reed and Allwright \cite{ReedA08} that a constant number of channels is not sufficient. 
Magnant and Martin strengthened their result that for any integer $k$ there exists a \emph{building} 
that requires exactly $k$ channels. In addition, their construction only requires floor-ceiling contacts.

We provide the first structural properties of this class. We first establish some links 
with the well-known notion of boxicity in Section~\ref{sec:boxicity}. Then in Section~\ref{sec:recognition}
we consider the recognition problem  and we prove that it is NP-complete to determine if a graph is $d$-CBU
for any integer $d\geq 3$. Then we provide a characterization in terms of an acyclic orientation of the class of general CBU.
Thanks to this characterization, it is immediate to realize that the class of CBU constitutes a proper 
sub-class of Hasse diagram graph (A Hasse diagram graph is the undirected 
graph obtained from a Hasse diagram associated to a poset).
Finally we prove in Section~\ref{sec:complexity} that several well studied optimization problems remains NP-hard
on either $2$- or $3$-CBU graphs. 

\section{Preliminaries}

We consider $\RR^d$ and $d$ orthogonal vectors $e_1,\ldots,e_d$ and we introduce a new class of geometric intersection graphs. 
Here, the vertices correspond to interior disjoint $d$-dimensional axis-parallel boxes in $\RR^d$, and two such boxes are only allowed to intersect on a $(d-1)$-dimensional box orthogonal to $e_1$. This class of graphs is denoted by \emph{$d$-CBU}, for \emph{Contact} graphs of $d$-dimensional \emph{Boxes} with \emph{Unidirectional} contacts. We denote \emph{CBU} the union of $d$-CBU for all $d$.

Note that 1-CBU correspond to the forests of paths.

\begin{claim}
For every $d\ge 1$, $d$-CBU graphs are triangle-free.
\end{claim}
Indeed, note that orienting the edges according to vector $e_1$ and labeling each arc with the coordinate of the corresponding $(d-1)$-hyperplane, one obtains an acyclic orientation such that for every vertex, all the outgoing arcs have the same label, all the ingoing arcs have the same label, and the label of ingoing arcs is smaller than the label of outgoing arcs. We call such a labeling of the arcs an \emph{homogeneous arc labeling}. Note that an oriented cycle cannot admit such a labeling. A triangle $abc$ oriented acyclically is, up to automorphism, such that $d^+(a)=2$, $d^+(b)=1$, and $d^+(c)=0$. Now $ab$ and $ac$ should have the same label, such as $ac$ and $bc$, but $ab$ and $bc$ should be distinct, a contradiction. Thus a triangle cannot admit a homogeneous arc labeling. This completes the proof of the claim.

With similar arguments one obtains the following for short cycles.
See Figure~\ref{fig:GoodBadOrientation}.
\begin{claim}\label{cl:45-cycles}
For any homogeneous arc labeling of a graph $G$, its restriction to a short cycle is as follows.
\begin{itemize}
    \item For a 4-cycle, the orientation is either such that there are two sources and two sinks, or
    it is such that there is one source and one sink linked by two oriented paths of length 2.
    \item For a 5-cycle, the orientation is such that there is one source and one sink linked by two oriented paths, one of length 2 and one of length 3.
\end{itemize}
\end{claim}
\begin{figure}
\centering
\includegraphics{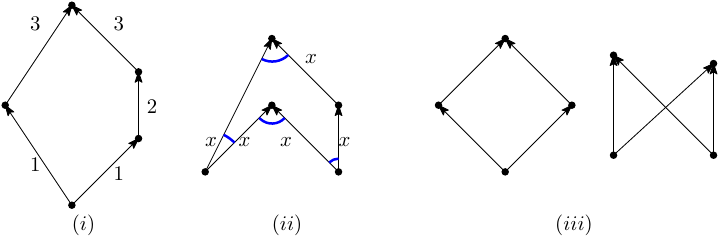}
\caption{$(i)$ Example of a good orientation of a $C_5$ with some valid 
labels, $(ii)$ example of bad orientation. Once a label $x$ is fixed for one arc, 
this label is propagated to all the arcs leading to the conclusion that $x < x$. $(iii)$ the two valid orientations of a $C_4$}
\label{fig:GoodBadOrientation}
\end{figure}

\section{Relation with Cover Graphs}
An undirected  graph is a \emph{cover graph} if it is the underlying graph of the Hasse diagram of some partial order.
It was shown by Brightwell \cite{Brightwell93} and also by  Ne\v{s}et\v{r}il and R\"odl~\cite{NesetrilR87,NesetrilR93,NesetrilR95}
that deciding whether a graph is a cover graph is NP-complete. 
However, they came up with a simple characterization in terms of acyclic orientations. 

Their characterization states that  a graph is a cover graph 
if and only if there exists an acyclic orientation without a quasi-cycle.
A \emph{quasi-cycle}, being an orientation of a cycle $(v_1,v_2,\ldots,v_n)$ with
the arcs $(v_i,v_{i+1})$ for all $1\leq n-1$ plus the arc $(v_1,v_n)$.
From a homogeneous arc labelling, the orientation provided by the labelling 
clearly fulfills the above defined condition. 
\begin{claim}\label{cl:large-cycles}
For any homogeneous arc labeling of a graph $G$, the orientation of $G$
does not contain any quasi-cycle.
\end{claim}
\begin{corollary}
The class of CBU graphs is contained in the class of cover graphs.
\end{corollary}

From the previous result, it is natural to ask whether both classes are equivalent.
The following remark provides the answer.

\begin{remark}
The class of CBU graphs is strictly contained in the class of cover graphs.
In Lemma \ref{lem-planar-not-CBU} we will exhibit a graph that is not CBU but is a cover graph.
\end{remark}
We will see in the following that an orientation of a triangle-free graph $G$,
fulfilling the conditions of Claim~\ref{cl:45-cycles} 
and of Claim~\ref{cl:large-cycles} may not admit a homogeneous arc labeling.

\section{Boxicity} \label{sec:boxicity}

The \emph{boxicity} $box(G)$ of a graph $G$, is the minimum dimension $d$ such that
$G$ admits an intersection representation with axis-aligned boxes. Of course, the graphs 
in $d$-CBU have boxicity at most $d$. The converse cannot hold for the graphs containing
a triangle, as those are not in CBU. However, some 
relations hold for triangle-free graphs. Let us begin with bipartite graphs.
\begin{theorem}\label{thm:box-CBU-bip}
Every bipartite graphs of boxicity $b$ belongs to $(b+1)$-CBU.
\end{theorem}
\begin{proof}
Consider a bipartite graph $G$ with vertex sets $A$ and $B$.
Consider a boxicity $b$ representation of $G$ and slightly expand each box in 
such a way that the intersection graph remains unchanged (we stop the expansion of 
the boxes before creating new intersections). Now, any two intersecting boxes intersect
on a $b$-dimensional box.
Assume that this representation is drawn in the space spanned by 
$e_2,\ldots,e_{b+1}$, and let us set for the first dimension (spanned by $e_1$) that 
the vertices of $A$ and $B$, correspond to the intervals $[0,1]$ and $[1,2]$, 
respectively. As $A$ and $B$ are independent sets, it is clear that the boxes in the representation are interior disjoint and that any two intersecting boxes intersect
on a $b$-dimensional box orthogonal to $e_1$. The obtained representation is thus 
a $(b+1)$-CBU representation of $G$.
\end{proof}

Theorem~\ref{thm:box-CBU-bip} does not extend to triangle-free graphs. We will see in the 
following section that there exists triangle-free graphs with bounded boxicity
that are not $d$-CBU, for any value $d$. Actually, Lemma~\ref{lem-planar-not-CBU}
tells that there exists such graphs with girth 5.
In other words, for a $3\le g\le 5$, there is no function $f_g$ such 
that every graph $G$, of girth at least $g$ and of boxicity $b$ belongs to $(f_g(b))$-CBU. 

\begin{problem}
For $g\ge 6$, is there a function $f_g$ such 
that every graph $G$, of girth at least $g$ and of boxicity $b$ belongs to $(f_g(b))$-CBU?
\end{problem}
By Theorem~\ref{thm:series-parallel}, we know that if $f_6$ exists, then $f_6(2)$ is at 
least 3.
Nevertheless,
the following theorem shows that subdividing the edges enables to consider every 
triangle-free graph.
An intersection representation is said \emph{proper} if two objects intersect if and only if 
some point of the representation belongs to these 2 objects, only.

\begin{theorem}\label{thm:box-proper-CBU}
For every graph $G$ having a proper intersection representation with axis-parallel boxes in 
$\RR^b$, the 1-subdivision of $G$ belongs to $(b+1)$-CBU.
\end{theorem}
\begin{proof}
Consider such a representation of $G$ and slightly expand each box in 
such a way that the intersection graph remains unchanged, and any two 
intersecting boxes intersect on a $b$-dimensional box.
Assume that this representation is drawn in the space spanned by 
$e_2,\ldots,e_{b+1}$, and for the first dimension (spanned by $e_1$) 
let us consider any vertex ordering, $v_1,\ldots,v_n$.
For the first dimension, a vertex $v_i$, corresponds to the interval $[2i,2i+1]$.
Clearly, none of these boxes intersect. 
Let us now add the boxes for the vertices added by subdividing 
the edges of $G$. For any edge $v_iv_j$, in the space spanned by $e_2,\ldots,e_{b+1}$, 
the expansion ensured that the intersection of $v_i$ and $v_j$ contains a box $B_{ij}$, 
that does not intersect any other box of the representation.
If $i<j$, the subdivision vertex of $v_iv_j$, is represented by $[2i+1,2j]\times B_{i,j}$. 
The obtained representation is clearly a $(b+1)$-CBU representation of the subdivision of $G$.
\end{proof}

\begin{corollary}\label{cor:box-CBU}
For every triangle-free graph $G$ of boxicity $b$, the 1-subdivision of $G$ 
belongs to $(b+1)$-CBU.
\end{corollary}

\section{Planar graphs}

While planar graphs have boxicity at most 3~\cite{T86,FelsnerMathews,cuboidsGD12}, 
many subclasses of planar graphs are known to have boxicity at most 2. 
This is the case for 4-connected planar graphs~\cite{T84}, and their subgraphs. 
The subgraphs of 4-connected graphs include every triangle-free planar graph 
(see Lemma 4.1 in~\cite{Gon-L}). As observed earlier, for those the representation is necessarily 
proper. 
For general planar graphs, the representation in $\RR^3$ provided in~\cite{FelsnerMathews} 
is clearly proper. So Theorem~\ref{thm:box-proper-CBU} implies the following.
\begin{corollary}\label{cor:planar-subd}
For every planar graph $G$, the 1-subdivision of $G$ 
belongs to $4$-CBU. Furthermore, if $G$ is triangle-free
then it even belongs to $3$-CBU.
\end{corollary}

\subsection{2-CBU graphs}

Given a 2-CBU representation of a graph $G$, and a vertical line $\ell$, the \emph{top box} of this representation with respect to $\ell$ is the highest box intersecting $\ell$. Now, the \emph{top sequence} of a 2-CBU representation is the sequence of top boxes obtained when parsing the representation with $\ell$ from left to right (see Figure~\ref{fig:2cbu-top}).

\begin{figure}
    \centering
    \includegraphics[width=0.6\textwidth]{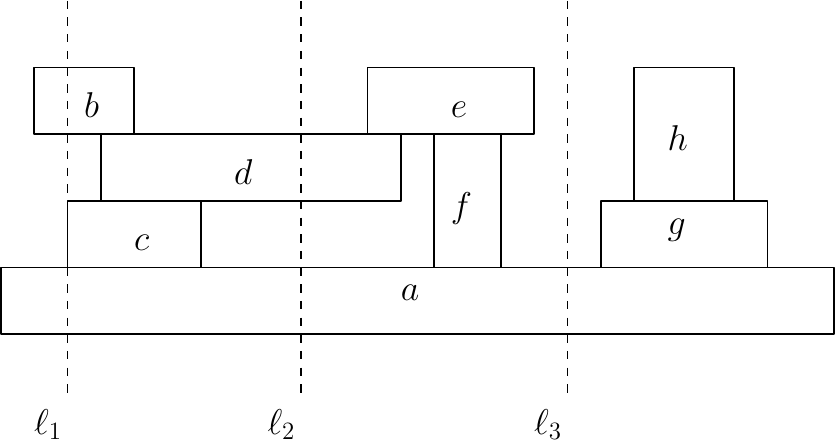}
    \caption{The top boxes with respect to $\ell_1$, $\ell_2$, and $\ell_3$ are $b$, $d$, and $a$, respectively. The top sequence of this 2-CBU representation is $a,b,d,e,a,g,h,g,a$.}
    \label{fig:2cbu-top}
\end{figure}

One can easily see that 2-CBU graphs are planar graphs, and that every forest 
is a 2-CBU graph. Actually, this class contains every triangle-free outerplanar graph.

\begin{theorem}
Every triangle-free outerplanar graph is 2-CBU.
\end{theorem}
\begin{proof}
Let us prove that for any connected outerplanar graph $G$, and any facial walk $v_1,v_2,\ldots, v_k, v_{k+1}=v_1$ on the outerboundary of $G$ (with separating vertices appearing several times in this walk), there exists a 2-CBU representation of $G$ with top sequence $v_1,v_2,\ldots, v_k,v_{k+1}$.

We proceed by induction on the number of vertices in $G$.
The statement clearly holds if $G$ has only one vertex.
A connected triangle-free outerplanar graph with more vertices contains either a vertex $v_i$ of degree one, or a cycle $v_i,\ldots,v_j$ of length at least four (i.e. $j-i\ge 3$), whose vertices $v_{i+1},\ldots, v_{j-1}$ have degree two in $G$.

\begin{figure}
    \centering
    \includegraphics[width=0.75\textwidth]{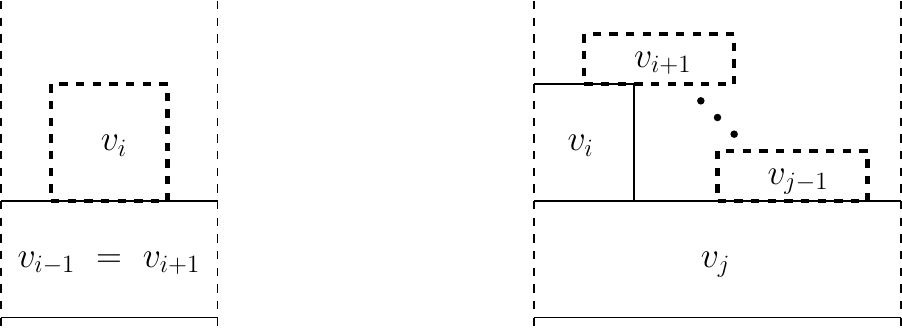}
    \caption{Left part : Adding a degree one vertex $v_i$ in the representation. This is done in the vertical stripe where the neighbor of $v_i$, $v_{i-1}=v_{i+1}$, is the top box. Right part : Adding $v_{i+1},\ldots,v_{j-1}$. This is done in the vertical stripe where $v_i$ and $v_{j}$ are the top box, successively.}
    \label{fig:outer}
\end{figure}  

In the first case we can add the box of $v_i$ in the representation of $G\setminus v_i$ obtained by induction (see Figure~\ref{fig:outer}, left). In the second case we consider the representation of $G\setminus \{v_{i+1},\ldots,v_{j-1}\}$ obtained by induction. Note that $v_i$ and $v_j$ now appear consecutively in its outerboundary, such as in the top sequence. It is thus easy to add the boxes of $v_{i+1},\ldots,v_{j-1}$ in the representation (see Figure~\ref{fig:outer}, right).
\end{proof}

\begin{figure}
\centering
\includegraphics[width=\textwidth]{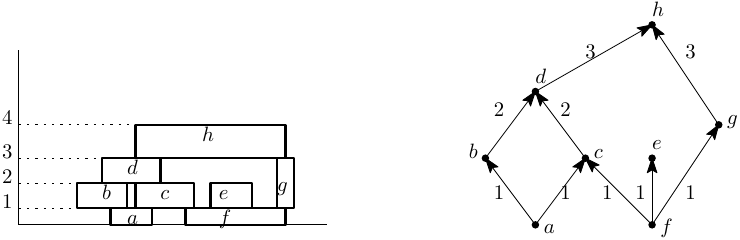}
\caption{An example of $2$-CBU graph and its associated acyclic orientation}
\label{fig:my_label}
\end{figure}

\begin{figure}
\centering
\includegraphics[width=0.45\textwidth,page=1]{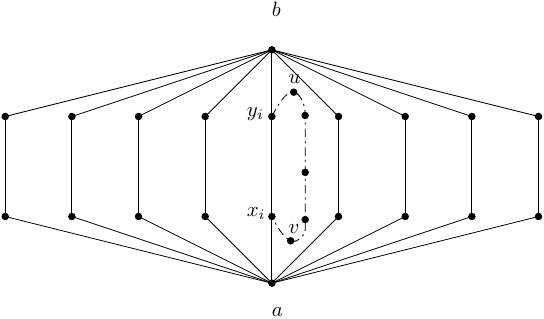}
\includegraphics[width=0.54\textwidth,page=2]{SeriesParallelNotCBU}
\caption{The series-parallel graph $G$ of Theorem~\ref{thm:series-parallel}.}
\label{fig:series-parallel}
\end{figure}

\begin{theorem}~\label{thm:series-parallel}
There are series-parallel graphs of girth 6 that are not 2-CBU. 
\end{theorem}
\begin{proof}
Consider a graph $G$ with two vertices, $a$ and $b$, linked by 9 disjoint $ab$-paths 
of length three $ax_iy_ib$, for $i\in\{1,\ldots,9\}$. Then for each edge $x_iy_i$ add
a length 5 path from $x_i$ to $y_i$. The obtained graph has girth 6 and is 
series-parallel (see Figure~\ref{fig:series-parallel}).

Note that in a 2-CBU representation of $G$, the box of $a$ (resp. $b$) has at most four neighbors such that their intersection contains a corner of $a$  (resp. $b$). Thus, there 
exists an $i\in\{1,\ldots,9\}$ such that one side of $x_i$ is contained 
in one side of $a$, and one side of $y_i$ is contained in one side of $b$. 
Now, whatever the way $x_i$ and $y_i$ intersect (a side of $x_i$ may be contained 
in a side of $y_i$, or the other way around, or also their intersection may contain 
a corner of each box), it is not possible to have the length 5 $x_iy_i$-path 
(see Figure~\ref{fig:series-parallel}, right). If 
a side of $x_i$ is contained in a side of $y_i$ there is no place left around $x_i$ 
to draw a third neighbor. If the intersection of $x_i$ and $y_i$ contains a corner 
of $x_i$ and a corner of $y_i$, there is space to draw a third neighbor for these 
vertices, say $u$ and $v$ respectively, but in that case the $uv$-path should go around
$a$ or $b$, but it would intersect the paths $ax_jy_jb$ with $j\neq i$. Thus $G$ does not
admit a 2-CBU representation.
\end{proof}

\begin{problem}
Is there a girth $g$ such that every series parallel graph of girth at least $g$ belongs to 2-CBU ?
\end{problem}

For planar graphs, the following theorem shows that such a bound on the girth does not exist.
Let us denote by $W^2_g$ the double wheel graph, obtained from a cycle $C_g$ by adding two 
non-adjacent vertices, each of them being adjacent to every vertex of $C_g$.
An edge incident to one of these two vertices (i.e., an edge not contained in $C_g$)
is called a \emph{ray}. Now, let $W'_g$ be the graph obtained from $W^2_g$ by subdividing $\lfloor g/2 \rfloor$ times every ray (see Figure~\ref{fig:W'6}). This graph is planar and has girth $g$.
\begin{theorem}~\label{thm:W'_g}
The graph $W'_g$ does not belong to 2-CBU. 
\end{theorem}
\begin{proof}
For any 2-CBU representation of the cycle $C$ of length $g$ there is a rectangle $R$, for example the one with the leftmost right side, such that none of the top or bottom side of $R$ is incident to the inner region. It is thus impossible to connect $R$ with a ray in the inner region. On the other hand, there is no planar embedding of $W'_g$ where $C$ bounds an inner face.
\end{proof}
\begin{figure}
    \centering
    \includegraphics{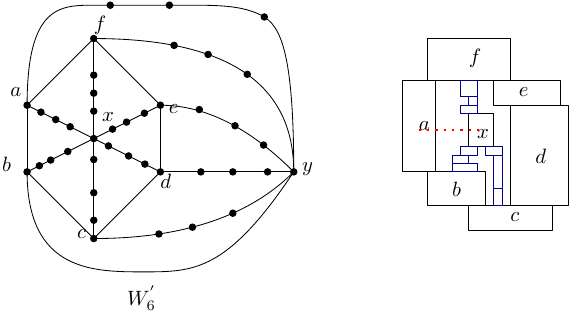}
    \caption{The graph $W'_6$.}
    \label{fig:W'6}
\end{figure}

\subsection{3-CBU graphs}

Bipartite planar graphs are known to be contact graphs of axis-aligned segments
in $\mathbb{R}^2$~\cite{Ben91,Fraysseix94}, and their boxicity is thus at most two. 
By  Theorem~\ref{thm:box-CBU-bip}, we thus have the following.

\begin{corollary}\label{cor:bip-planar}
Every bipartite planar graph belongs to 3-CBU.
\end{corollary}
The following lemma tells us that this property does not generalize, in a strong sense, 
to triangle-free planar graphs.

\begin{lemma}~\label{lem-planar-not-CBU}
There exists a girth 4 planar graph that is not CBU.
\end{lemma}

\begin{figure}
\centering
\includegraphics[page=1]{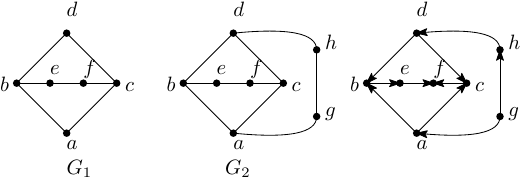}
\includegraphics[page=2]{PlanarNonCBU-2}
\includegraphics[page=3]{PlanarNonCBU-2.pdf}
\caption{The graph $G_3$ is planar and non CBU. It is however a cover graph. }
\label{fig:PlanarNonCBU}
\end{figure}

\begin{proof}
Let us consider the graph $G_1$ represented in Figure \ref{fig:PlanarNonCBU}. 
One can show that this graph does not admit any valid CBU orientation where 
in the $C_4$ induced by $a,b,c$ and $d$, $a$ is a source and $d$ is a sink (nor
the converse). Let us assume that there exists a CBU orientation 
such that $a$ is a source and $d$ is a sink. By fixing the orientation 
of edges $ab$ and $ac$ from $a$ to $b$ and from $a$ to $c$ respectively, and by Claim~\ref{cl:45-cycles}, the edges $be$ and $cf$ have to be oriented from $b$ to $e$ and from 
$c$ to $f$. The edge $ef$ is oriented in any direction, \emph{w.l.o.g.} 
let us say from $e$ to $f$. But in that case, in the $C_5$ induced by $b,e,f,c$ and 
$d$ contains two sources and two sinks, which is not a valid CBU orientation (By Claim~\ref{cl:45-cycles}).

By adding a path of length $3$ between $a$ and $d$, we obtain the graph $G_2$.
From the previous observation, we can conclude the same property 
for vertices $b$ and $c$. Hence, in the valid orientation of $G_2$, $a$ and $d$ 
are sources and $b$ and $c$ are sinks (or the converse). 
Let us now consider the graph $G_3$ obtained by gluing two copies of $G_2$ in a special 
manner (see Figure \ref{fig:PlanarNonCBU}). Let us remark that the graph 
obtained is planar.
Let us consider, \emph{w.l.o.g.}, that $a$ and $d$ are sources and $b$ and $c$ are
sinks in the $C_4$ induced by $a,b,c,$ and $d$.
Recall that a valid orientation of a $C_5$ contains exactly one source and one sink.
Thus, in the $C_5$ induced by the vertices $a,c,d,k,$ and $j$, the vertex $c$ has to be
a sink from the already fixed orientation. Hence, it forces the edge $aj$ to be oriented 
from $j$ to $a$ and the edge $dk$ from $k$ to $d$ (the edge $k,j$ can be oriented 
in any direction), since the length of a path from a source to a sink in an orientation 
is exactly $2$ for one path and $3$ for the other. 
 
Then in the partial orientation obtained, we can conclude that in the $C_4$ induced 
by $a,c,i,$ and $j$, the vertex $c$ will be a sink and vertex $j$ will be a source.
However, as mentioned in the beginning of this proof, this orientation will not 
lead to valid orientation, since $j$ and $c$ play the same role 
as $a$ and $d$ in $G_1$. Hence, $G_3$ does not admit a valid CBU orientation.
\end{proof}

\begin{remark}
The graph $G_3$ used in the proof of Lemma \ref{lem-planar-not-CBU} 
is actually a cover graph. In Figure \ref{fig:PlanarNonCBU}, the bottom picture depicts 
its Hasse diagram.
\end{remark}

Since every planar graph with girth at least 10 has circular chromatic number at most $5/2$~\cite{girth10}, the forthcoming Theorem~\ref{thm:CBU-chi-c} implies that such a graph necessarily belongs to CBU.
\begin{problem}
What is the lowest $g\in [5,\ldots,10]$ such that every planar graph $G$ with girth at least $g$ belongs to CBU.
\end{problem}

\section{Structural properties of {$d$}-CBU and CBU}

\begin{theorem}\label{thm:CBU-d-d+1}
For every $d\ge 1$, the class of $d$-CBU graphs is strictly contained in the class of $(d+1)$-CBU graphs.
\end{theorem}
The case $d=1$ follows from the earlier observation that 1-CBU graphs correspond to forests of paths, 
and from the many examples of 2-CBU graphs provided above. 
For $d\ge 2$, the following structural lemma allows us to
translate the strict containment of boxicity $b$ bipartite graphs, 
to the strict containment of $d$-CBU. Indeed, it is known that the graph obtained from the complete
bipartite graph $K_{2b,2b}$ by removing a perfect matching has boxicity exactly $b$~\cite{CDS09}.

\begin{lemma}\label{lem:bip-exact-CBU}
Given a connected bipartite graph $B$, with parts $X$ and $Y$, let $B'$ be the graph
obtained from $B$, by adding a path $xzy$ and by connecting $x$ and $y$ to every 
vertex in $X$ and $Y$, respectively. Then, $B$ has boxicity at most $d$ if and only if
$B'$ belongs to $(d+1)$-CBU.
\end{lemma}
\begin{proof}
Let us begin with the simpler "only if" part. We proceed as in the proof of 
Theorem~\ref{thm:box-CBU-bip} in order to obtain a $(d+1)$-CBU representation of $B$ 
such that every vertex of $X$ (resp. $Y$) corresponds to $[0,1]$ (resp. $[1,2]$) 
in the space spanned by $e_1$. Then it suffices to add the boxes for $x,y$ and $z$.
For a sufficiently large $\Omega$, 
$x$ is represented by $[-1,0]\times [-\Omega,+\Omega]\times \ldots \times [-\Omega,+\Omega]$,
$y$ is represented by $[2,3]\times [-\Omega,+\Omega]\times \ldots \times [-\Omega,+\Omega]$, and
$z$ is represented by $[0,2]\times [\Omega-1,\Omega]\times \ldots \times [\Omega-1,\Omega]$.

For the "if" part, consider a $(d+1)$-CBU representation of $B'$, and the homogeneous arc labeling
of $B'$ induced by this representation. We first prove that all the arcs between $X$ and $Y$ are oriented in 
the same direction. Towards a contradiction, consider a path $x_1y_2x_3$ with $x_1,x_3\in X$ and $y_2\in Y$,
and such that the edges are oriented from $x_1$ to $y_2$, and from $y_2$ to $x_3$. This forces the remaining 
edges of the 4-cycle $yx_1y_2x_3$ to be oriented from $x_1$ to $y$, and from $y$ to $x_3$
(see Claim~\ref{cl:45-cycles}). Now we cannot orient the edge $y_2x, xz, zy$ in such a way to fulfill
Claim~\ref{cl:45-cycles} for the 5-cycles $xzyx_1y_2$ and $xzyx_3y_2$. Indeed for the first one, $yz$ 
should be oriented from $y$ to $z$, while for the second one it should be oriented from $z$ to $y$,
a contradiction.

This orientation ensures that the labels of all the arcs is the same. This implies that there is an hyperplane
$\HH$ orthogonal to $e_1$ such that for any pair of intersecting 
boxes $x'\in X$ and $y'\in Y$, their intersection belongs to $\HH$. 
This implies that projecting the $(d+1)$-CBU representation (restricted to $B$)
along $e_1$ leads to a boxicity $d$ representation of $B$.
\end{proof}

It is clear that CBU is hereditary (i.e. closed under induced subgraphs) 
but actually it is also closed under subgraphs.
\begin{theorem}~\label{thm:CBU-subgraph}
For any subgraph $H$ of $G$,  $G\in$ CBU implies that $H\in$ CBU. More precisely, 
if there is a complete bipartite graph $K_{a,b}$ such that $V(K_{a,b}) \subseteq V(G)$, 
and such that $E(H) = E(G) \setminus E(K_{a,b})$,
then if $G$ belongs to $d$-CBU then $H$ belongs to $(d+1)$-CBU.
\end{theorem}
\begin{proof}
Let $A,B$ be the parts of  $K_{a,b}$. Given a CBU representation of $G$ in $\RR^d$ we are going to build a CBU representation of $H$ in $\RR^{d+1}$. For this, the first $d$ intervals defining each $d$-box remain unchanged while the last interval is $[0,1]$ for the vertices in $A$, $[2,3]$ for the vertices in $B$, and $[0,3]$ for the remaining vertices. It is now easy to check that two boxes intersect if and only if they intersect in $G$ and if they are not adjacent in  $K_{a,b}$. It is also clear that the intersections occur on planes orthogonal to $e_1$.
\end{proof}

The graph class CBU is also closed by the addition of false twins.
\begin{theorem}~\label{thm:CBU-twin}
For any graph $G$ and any vertex $v$ of $G$, consider the graph $G^v$ obtained from $G$ 
by adding a new vertex $v'$ such that $N(v') = N(v)$. Then $G\in$ CBU if and only if 
$G^v\in$ CBU. Furthermore, if $G\in d$-CBU then $G^v\in (d+1)$-CBU.
\end{theorem}
\begin{proof}
The "if"  part is obvious as $G$ is an induced subgraph of $G^v$.
For the "only if" part, given a CBU representation of $G$ in $\RR^d$ we are going to build a CBU representation of $G^v$ in $\RR^{d+1}$. For this, the first $d$ intervals defining each $d$-box remain unchanged, and those of $v'$ are the same as 
those of $v$. The last interval is $[0,1]$ for $v$, $[2,3]$ for $v'$, and $[0,3]$ for all the remaining vertices. It is now easy to check that two boxes intersect if and only if they intersected and if one of them is distinct from $v$ or $v'$. It is also clear that the intersections occur on planes orthogonal to $e_1$.
\end{proof}

\emph{Shift graphs} were introduced by P. Erd\H{o}s and A. Hajnal in~\cite{EH64} 
(see Theorem 6 therein). Those are the graphs $H_m$ whose vertices are the ordered pairs 
$(i,j)$ satisfying $1\le i < j\le m$, and where two pairs $(i,j$) and $(k,l)$
form an edge if and only if $j=k$ or $l=i$.
Note that such graphs admit a homogeneous arc labeling $\ell$ defined by 
$\ell(\{(i,j),(j,k)\}) = j$, and by orienting any edge $\{(i,j),(j,k)\}$ from $(i,j)$ to $(j,k)$.

\begin{theorem}~\label{thm:D_n}
The graph $H_m$ belongs to $(m-1)$-CBU. Furthermore, $H_m$ has a CBU representation such 
that in the first dimension the vertex $(i,j)$ corresponds to interval $[i,j]$.
\end{theorem}
\begin{proof}
This clearly holds for the one vertex graph $H_2$.
By induction on $m$ consider a representation of $H_{m-1}$, add a false twin for every vertex 
$(i,m-1)$ and modify the first interval of these new twins, so that the interval $[i,m-1]$ becomes 
$[i,m]$. These boxes correspond to the vertices $(i,m)$  with $i<m-1$.
For the vertex $(m-1,m)$, one should add a box $[m-1,m]\times [-\Omega,+\Omega]\times\ldots \times [-\Omega,+\Omega]$, for a sufficiently large $\Omega$. To deal with the intersections between this box and the boxes of the other vertices $(i,m)$, we add a new dimension such that vertex 
$(m-1,m)$ has interval $[1,2]$, the vertices $(i,m-1)$ have interval $[1,2]$, the vertices $(i,m)$ with $i<m-1$ have interval $[3,4]$, and all the other vertices have interval $[1,4]$.
\end{proof}

\begin{theorem}~\label{thm:CBU_as_SHIFT}
For every $n$-vertex graph $G$ the following properties are equivalent.
\begin{itemize}
    \item[a)] $G$ belongs to CBU.
    \item[b)] $G$ admits a homogeneous arc labeling.
    \item[c)] $G$ is the subgraph of a graph $H_m^t$, obtained from the shift graph $H_m$ by iteratively adding $t$ false twins, for some values $m,t$ such that $m+t\le n+1$.
    \item[d)] $G$ belongs to $(2n-1)$-CBU.
\end{itemize}
\end{theorem}
\begin{proof}
We have already seen that $a)\Rightarrow b)$. Let us show $b)\Rightarrow c)$.
Consider a homogeneous arc labeling of $G$, with labels in $[2,m-1]$, for the minimum $m$.
By minimality of $m$, note that all the labels are used, and thus $m-2\le n-1$. 
Let $H_m^t$ be the graph obtained from 
the shift graph $H_m$ by adding $t_{i,j}$ false twins of vertex $(i,j)$ if there are $t_{i,j} +1$
vertices of $G$ whose incoming arcs are labeled $i$, and whose outgoing arcs are labeled $j$.
For the vertices without incoming (resp. outgoing) arcs assume that those are labeled 1 (resp. $m$).
Consider now an injective mapping $\gamma\ :\ V(G) \longrightarrow V(H_m^t)$, such that any vertex 
with incoming and outgoing arcs labeled $i,j$ is mapped to $(i,j)$ or one of its twins.
This mapping ensures us that $G$ is a subgraph of $H_m^t$. Indeed, for any two 
adjacent vertices $u,v$ of $G$ linked by an edge labelled $j$ oriented from $u$ to  $v$, 
their incoming and outgoing arcs are labeled $i,j$ and $j,k$ respectively, for some $i <j<k$,
and thus the vertices $\gamma(u)$ and $\gamma(v)$ of $H_m^t$ are adjacent, as they
correspond to or are twins of $(i,j)$ and $(j,k)$.

We now show $c)\Rightarrow d)$. Consider a graph $H_m^t$ containing $G$ as a subgraph, 
for some $m,t$ such that $m+t\le n+1$.
By Theorem~\ref{thm:D_n} and Theorem~\ref{thm:CBU-twin} we have that  $H_m^t$ belongs to 
$(m-1+t)$-CBU, and so to $n$-CBU. Starting from $H_m^t$ one can obtain $G$ by successively deleting 
$n-1$ stars $K_{1,b}$, so by Theorem~\ref{thm:CBU-subgraph}, we have that $G$ belongs to $(2n-1)$-CBU.
Finally, $d)\Rightarrow a)$ is obvious.
\end{proof}

It is easy to see that every complete bipartite graph belongs to 3-CBU. By  Theorem~\ref{thm:CBU-subgraph}, removing stars  $K_{1,b}$
centered on the smallest part, one obtains that every $n$-vertex bipartite graph 
belongs to $(\lfloor n/2 \rfloor +3)$-CBU. One can reach a slightly better bound
from Theorem~\ref{thm:box-CBU-bip}, 
and the fact that for every graph $G$, $box(G)\le \lfloor n/2\rfloor$~\cite{Roberts69}.
\begin{corollary}
Every bipartite graph $G$ belongs to CBU. Furthemore, 
if $|V(G)| = n$ then $G$ belongs to $(\lfloor n/2 \rfloor +1)$-CBU.
\end{corollary}
As already mentioned, some bipartite graphs have arbitrary large boxicity, and thus 
there is no fixed $d$ such that every bipartite graph belongs to $d$-CBU.
For large girth graphs it is a different.
\begin{theorem}\label{thm:not-CBU-girth}
For any $g\ge 3$, there exist graphs of girth $g$ not contained in CBU.
\end{theorem}
\begin{proof}
Indeed, for any $g\ge 3$ there exist graphs of girth $g$ with 
fractional chromatic number at least 4~\cite{Erdos1959}. (Actually, their
fractional chromatic number is arbitrarily large). By Theorem~\ref{thm:chi_f},
such graphs cannot belong to CBU.
\end{proof}
Nevertheless, the following remains open.
\begin{problem}
Are there integers $d, g$ such that every girth $g$ graph $G$ of CBU belongs to $d$-CBU?
\end{problem}

The remarks above imply that testing if a bipartite graph belongs to CBU is obvious (computable 
in constant time), while for girth $g$ graphs the question is more involved, as CBU has such 
graphs included and some other excluded. The following section treats the computational problem
of recognizing CBU graphs.

\section{Recognition}\label{sec:recognition}

Computing the boxicity of a bipartite graph is a difficult problem.
It is known that deciding whether a bipartite graph has boxicity two is NP-complete~\cite{Kra94}.
Furthermore, it is proven in~\cite{AdigaBC10} that it is not possible to approximate the 
boxicity of a bipartite graph within a $O(n^{0.5-\varepsilon})$-factor in polynomial time, unless $NP=ZPP$.
By Lemma~\ref{lem:bip-exact-CBU}, for every bipartite graph $B$ there is a graph $B'$ 
(obtained in polynomial time) such that the minimum value $d$ for which $B'$ belongs to $d$-CBU
is exactly $d=box(B)+1$.
\begin{corollary}\label{cor:3CBU-NPc}
It is NP-complete to decide whether a graph belongs to 3-CBU. Furthermore, 
unless $NP=ZPP$, one cannot approximate in polynomial time and within a 
$O(n^{0.5-\varepsilon})$-factor, the minimum value $d$ for which an input graph $G$ belongs to $d$-CBU.
\end{corollary}
This implies that for most values $d$ the problem of deciding whether an input graph belongs to $d$-CBU,
cannot be computed in polynomial time, unless $NP=ZPP$. 
The hypothesis $NP = P$ being stronger than $NP=ZPP$, it would be stronger to know that 
it is NP-complete to decide if an input graph belongs to $d$-CBU.

\begin{problem}\label{pb:dCBU-recogn}
For which values $d$, is it NP-complete to decide whether a graph belongs to $d$-CBU?
Are there values $d$, in particular for $d=2$, for which the problem is polynomial?
\end{problem}
By Lemma~\ref{lem:bip-exact-CBU}, this problem would be solved, for $d\ge 3$, if the following problem admits a positive answer.
\begin{problem}\label{pb:boxicity-poly}
For any $d\ge 3$, is it NP-complete to decide whether a bipartite graph $B$ has boxicity at most $d$?
\end{problem}
Another computational problem is testing the membership in CBU.
\begin{problem}\label{pb:CBU-poly}
Is it polynomial to decide whether a graph belongs to CBU?
\end{problem}

We have seen that some triangle-free planar graphs, or some graphs with arbitrary large girth, 
are not in CBU. We can thus restrict the problem.
\begin{problem}\label{pb:CBU-poly-bis}
Is it polynomial to decide whether a planar graph $G$ belongs to CBU? For some $g\ge 3$, 
is it polynomial to decide whether a graph $G$ of girth at least $g$ belongs to CBU?
\end{problem}

\subsection{Recognition through forbidden induced subgraphs}

As CBU and $d$-CBU are closed under induced subgraphs, they are characterized by a set
of minimal excluded induced subgraphs, $\mathcal{F}_{CBU}$ and $\mathcal{F}_{d-CBU}$.
If one of these sets is finite, then recognizing the corresponding class becomes polynomial-time
tractable.
Thus by Corollary~\ref{cor:3CBU-NPc}, the set $\mathcal{F}_{3-CBU}$ 
(resp. $\mathcal{F}_{d-CBU}$ for $d\ge 4$) is not finite, unless $P=NP$ (resp. unless $NP=ZPP$). 
For the set $\mathcal{F}_{2-CBU}$ (resp. $\mathcal{F}_{CBU}$), we are sure that it is infinite.
Indeed, Theorem~\ref{thm:W'_g} (resp. Theorem~\ref{thm:not-CBU-girth}) provides an infinite 
sequence of graphs $\left(G_i\right)_{i\ge 0}$ not in 2-CBU (resp. not in CBU) such that the girth 
of $G_i$ is at least $i$. If there was an $n$ such that every graph in $\mathcal{F}_{2-CBU}$ (resp. $\mathcal{F}_{CBU}$) has at most $n$ vertices, then to exclude $G_{n+1}$ one would need to have a 
tree in $\mathcal{F}_{2-CBU}$ (resp. $\mathcal{F}_{CBU}$). This is not the case as for every tree $T$,
we have that $T\in 2$-$CBU \subseteq CBU$.

\subsection{Recognition through homogeneous arc labelings}

By Theorem~\ref{thm:CBU_as_SHIFT}, a graph $G$ belongs to CBU
if and only if it admits a homogeneous arc labeling. 
If we are given an orientation of a graph $G$ it is simple to check whether this orientation admits such labeling. 
For example, one can use linear programming. For each arc $uv$, set a variable $\ell_{uv}$ corresponding to a label, and for any two incident arcs, add a constraint. 
For two arcs $uv$ and $uw$ (resp. $uv$ and $wv$), the constraint is $\ell_{uv} = \ell_{uw}$ (resp. $\ell_{uv} = \ell_{wv}$). 
For two arcs $uv$ and $vw$, the constraint is $\ell_{uv}+1 \le \ell_{vw}$.
Problem~\ref{pb:CBU-poly} thus reduces to deciding whether a graph $G$ admits an orientation that 
is homogeneously labelable. In the following we characterize such orientations.

A cycle $(v_0,v_1,\ldots,v_{n-1})$ is said \emph{badly oriented} if there is a vertex $v_i$ 
whose incident arcs are $v_{i-1}v_i$ and $v_iv_{i+1}$, and if there is no vertex $v_j$
whose incident arcs are $v_{j+1}v_j$ and $v_jv_{j-1}$ (indices being considered $\bmod\ n$).
\begin{theorem}\label{thm:labelable-orientation}
An orientation of a graph $G$ admits a homogeneous labeling if and only if there is no badly oriented cycle.
\end{theorem}
\begin{proof}
For the "only if" part, consider a badly oriented cycle  $(v_0,v_1,\ldots,v_{n-1})$ 
with arcs $v_{n-1}v_0$ and $v_0v_{1}$, but with no vertex $v_j\neq v_0$
whose incident arcs are $v_{j+1}v_j$ and $v_jv_{j-1}$. This latter condition
implies that in any homogeneous labeling the sequence of labels for the edges (without considering their 
orientation) $v_0v_1, v_1v_2, \ldots ,v_{n-2}v_{n-1},v_{n-1}v_0$ is non-decreasing, while the former 
condition implies that the label of $v_0v_1$ is greater than the one of $v_{n-1}v_0$, a contradiction.
Thus this orientation of $G$ does not allow any homogeneous labeling.

For the "if" part, consider a graph $G$ oriented without badly oriented cycle, and 
consider a source $u$, and let us denote $v_1,\ldots,v_n$ its out-neighbors. If for every vertex 
$v_i$, $u$ is its unique in-neighbor, then by recurrence on the number of vertices we assume
that $G\setminus \{u\}$ has a homogeneous labeling, and we label the arcs incident to $u$ 
with a sufficiently small value, say $-\Omega$. In that case it is easy to check that this
labeling is homogeneous. 

Otherwise, let $v_i$ and $u'$ be vertices such that $G$ has arcs from both $u$ and $u'$ 
toward vertex $v_i$. In that case, consider the oriented graph $G'$ obtained from 
$G\setminus \{u\}$ by adding the arcs $u'v_1, \ldots , u'v_n$, if missing. 
\begin{claim}
$G'$ has no badly oriented cycle.
\end{claim}
\begin{proof}
If $G'$ had a badly oriented cycle $C$, this one should go through a newly added 
arc $u'v_j$. If $v_i\notin C$, by replacing the arc $u'v_j$ by the path $(u',v_i, u, v_j)$ 
one would obtain a badly oriented cycle in $G$, a contradiction. We thus assume that
$v_i\notin C$, and now by replacing the arc $u'v_j$ by the path $(u',v_i, u, v_j)$ we
obtain a badly oriented closed walk $W$ (that is a walk where there are consecutive "forward" arcs, 
but no consecutive "backward" arcs). Let us denote $P$ and $P'$ the sub-paths of 
$C\setminus \{u'v_j\} \subsetneq G$ linking $v_i$ and $v_j$, and linking $u'$ and $v_i$, respectively. 

Let us show that if the edge incident to $v_i$ in $P'$ is oriented from $v_i$ to the other end, 
denoted $v$, then this arc is backward with respect to $C$. Indeed, the cycle $C_{P'}$ of $G$ 
formed by $P'$ and the arc $u'v_i$, has consecutive arcs oriented in the same direction, 
$u'v_i$ and $v_iv$, and (as $G$ contains no badly oriented cycles) 
has consecutive arcs oriented in the other direction. The latter pair of arcs belonging both 
to $P'\subset C$, they are forward with respect to $C$, thus $v_iv$ is backward.

Similarly, let us show that if the edge incident to $v_i$ in $P$ is oriented from $v_i$ to the 
other end, denoted $w$, then this arc is backward with respect to $C$. Indeed, they cycle $C_{P}$ 
of $G$ formed by $P$ and the arcs $uv_i$and $uv_j$, has consecutive arcs oriented in the same 
direction, $uv_i$ and $v_iw$, and (as $G$ contains no badly oriented cycles) 
has consecutive arcs oriented in the other direction. The latter pair of arcs belong both 
to $P\subset C$, or they are the arcs incident to $v_j$.
In the former case, these arcs are forward with respect to $C$, thus $v_iv$ is backward.
In the latter case, replacing $uv_j$ with $u'v_j$, one has that the incident arcs of $v_j$ in $C$
are oriented in the same direction. this direction is thus the forward direction, and in that case 
also $v_iv$ is backward.

We thus have that the arcs incident to $v_i$ cannot be oriented in the same direction 
(they would form consecutive backward arcs in $C$), and they are not both oriented from
$v_i$ to the other end (they would be both backwards although they have distinct directions).
Now we distinguish cases according to the position of the consecutive forward 
arcs in $C$. We have that:
\begin{itemize}
    \item[a)] there are two consecutive forward arcs in $P\cup\{u'v_j\}$, or
    \item[b)] there are two consecutive forward arcs in $P'\cup\{u'v_j\}$.
\end{itemize}
In case a), the cycle $C_P$ of $G$ has consecutive forward arcs (by replacing if necessary the arc 
$u'v_j$ with $uv_j$). Since this cycle is not badly oriented it also contains consecutive 
backward arcs. According to the orientation of the arcs, those backwards arcs cannot be the arcs 
incident to $u$, or those incident to $v_i$. Thus they belong both to $P\cup\{uv_j\}$, 
but this would imply that $C$ also contains consecutive backward arcs, a contradiction.

In case b), the cycle $C_{P'}$ of $G$ has consecutive forward arcs (by replacing if necessary the arc $u'v_j$ with $u'v_i$). Since this cycle is not badly oriented it also contains consecutive 
backward arcs. According to the orientation of the arcs, those backwards arcs cannot be the arcs incident to $v_i$. This would imply that $C$ also contains consecutive backward arcs, a contradiction.

This concludes the proof of the claim
\end{proof}

So now, by recurrence on the number of vertices we can assume
that $G'$ has a homogeneous labeling, and let $\ell$ be the label of the arcs outgoing from $u'$. 
In that case one can derive a labeling of $G$ by keeping the same labels, and by setting the
label $\ell$ for the arcs outgoing from $u$. It is easy to check that this labeling is homogeneous. 
\end{proof}

Note that Theorem~\ref{thm:labelable-orientation} provides another proof that CBU contains every bipartite graph. Indeed, orienting all the edges from one part toward the other, the direction of the arcs alternate along any cycle, and so there is no badly oriented cycle. Actually, we can go a little further.
\begin{theorem}\label{thm:CBU-chi-c}
Every graph $G$ with circular chromatic number $\chi_c(G) \le 5/2$ belongs to CBU.
\end{theorem}
\begin{proof}
A graph $G$ with circular chromatic number $\chi_c(G) \le 5/2$ has a
homomorphism into the circular complete graph $K_{5/2}$ that is the 5-cycle. 
As this graph belongs to CBU the theorem follows from Theorem~\ref{thm:CBU-homo}.
\end{proof}

Note that we cannot replace $5/2$ by $8/3$ in Theorem~\ref{thm:CBU-chi-c}, as one can easily check that 
every orientation of $K_{8/3}$ contains a badly oriented cycle.
\begin{problem}
What is the largest $c$ such that every graph $G$ with $\chi_c(G) \le c$ 
(or with $\chi_c(G) < c$) belongs to CBU.
\end{problem}
\begin{theorem}\label{thm:CBU-homo}
Given two graphs $G, H$ such that there is an homomorphism $\gamma\ :\ V(G) \longrightarrow V(H)$,
then if $H \in$ CBU we have that $G \in$ CBU.
\end{theorem}
\begin{proof}
By Theorem~\ref{thm:CBU_as_SHIFT}, the graph $H$ admits a homogeneous arc labeling, $\ell_H$.
Orient the edges of $G$ in such a way that $uv\in E(G)$ is oriented as the edge
$\gamma(u)\gamma(v)\in E(H)$, that is from $u$ to $v$ if and only if
$\gamma(u)\gamma(v)$ is oriented from $\gamma(u)$ to $\gamma(v)$ in $H$.
Similarly we copy the labeling of $H$'s arcs by setting $\ell_G(uv)=\ell_H(\gamma(u)\gamma(v))$. One can easily check that this is a homogeneous arc labeling of $G$,
and thus that $G$ belongs to CBU.
\end{proof}

\section{Chromatic Number and Independent Sets}

While 2-CBU graphs have chromatic number at most 3 (by Gr\"otzsch's theorem), 3-CBU graphs have unbounded chromatic number. 
\begin{theorem}[Magnant and Martin~\cite{MagnantM11}]
For any $\chi \ge 1$, there exists a graph in 3-CBU with chromatic number $\chi$.
\end{theorem}\label{thm:Magnant}
However, these graphs have bounded fractional chromatic number, and thus have 
linear size independent sets. Indeed, G. Simonyi and G. Tardos~\cite{ST11} showed that 
shift graphs have fractional chromatic number less than 4. 
As such a bound extends by adding a false twin and by taking a subgraph,
we have the following.

\begin{theorem}~\label{thm:chi_f}
For any graph $G\in$ CBU, $\chi_f(G)<4$, and $\alpha(G)> |V(G)|/4$.
\end{theorem}

For planar graphs in CBU, this bound on $\chi_f$ can be improved by one, but not more.
\begin{theorem}~\label{thm:chi_f_planar}
For every planar graph $G$ in CBU we have $\chi_f(G) \le \chi(G) \le 3$. On the other hand, for every $n\equiv 2\ (\bmod\ 3)$ there is a $n$-vertex planar graph $G$ in CBU such that $\alpha(G)=(n+1)/3$, and thus $\chi_f(G) \ge n/\alpha(G) = 3 - \frac{3}{n+1}$.
\end{theorem}
\begin{proof}
The first statement follows from Gr\"otzsch's theorem. The second statement follows from graphs constructed by Jones~\cite{JonesJCTB}, which were proved to have independence number 
$\alpha(G)=(n+1)/3$. Those graphs form a sequence $J_1,J_2,\ldots$ such that $J_1$ is the 5-cycle $(a_1,b_1,c_1,d,e)$, and such that 
$J_{i+1}$ is obtained from $J_i$ by adding three vertices $a_{i+1},b_{i+1},c_{i+1}$ such that $N(a_{i+1})=\{b_i,b_{i+1}\}$, $N(b_{i+1})=\{a_{i+1},c_{i+1}\}$, and $N(c_{i+1}) = \{a_i,c_i,b_{i+1}\}$
(see Figure~\ref{fig:Jones}). It is already known that those graphs are planar, and it does only remain to show that they belong to CBU. Let us do so by exhibiting a homogeneous arc labeling $\ell$. This labeling is such that for any $i\ge 1$ we orient the edges $a_ib_i$ and $b_ic_i$ toward $b_i$, we orient the edges $x_iy_{i+1}$, for $x,y\in\{a,b,c\}$, from $x_i$ towards $y_{i+1}$, and we set $\ell(a_ib_i) = \ell(a_ic_{i+1)} = 2i$, $\ell(c_ib_i)=\ell(c_ic_{i+1)}=2i$, and $\ell(b_ia_{i+1)}=2i+1$.
By examining Figure~\ref{fig:Jones} it is clear that this is a homogeneous arc labeling.
\end{proof}
\begin{figure}
    \centering
    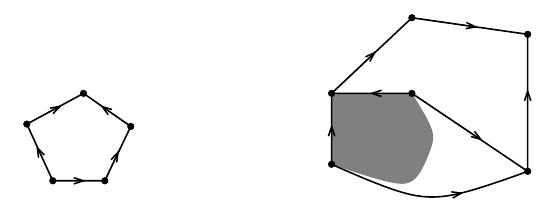
    \includegraphics{Jones}
    \caption{The Jones graphs $J_1$ and $J_{i+1}$, with a homogeneous arc labeling. For every $i\ge 1$, this embedding is such that the path $a_ib_ic_i$ is on the outer-boundary. Thus, adding vertices $a_{i+1},b_{i+1},c_{i+1}$ does not break planarity.}
    \label{fig:Jones}
\end{figure}

Although $2$-CBU lies in the intersection of CBU and planar graphs, it might be the case that the fractional chromatic number of graphs in $2$-CBU is bounded by some $c < 3$. Indeed, Jones graphs $J_i$, for a sufficiently large $i$, seem to not be in $2$-CBU.
\begin{problem}
Is there a $c<3$ such that every graph $G$ in $2$-CBU has fractional chromatic number $\chi_f(G) \le c$ ?
\end{problem}
A positive answer to this question, would give support to two conjectures. Let $\mathcal{P}_{g\ge 5}$ be the set of planar graph with girth at least five, and let $\mathcal{P}^f_{g\ge 4}$ be the set of planar graph with girth at least four, where every 4-cycle bounds a face.
Clearly $\mathcal{P}_{g\ge 5} \subsetneq \mathcal{P}^f_{g\ge 4}$, since these classes avoid Jones graphs it is conjectured that graphs in $\mathcal{P}_{g\ge 5}$, or more generally graphs in $\mathcal{P}^f_{g\ge 4}$, have fractional chromatic number at most $c$, for some $c<3$~\cite{conj-frac-2,conj-frac-1}. However, our problem is not a sub-case of these conjectures (as $K_{2,t}$ belongs to $2$-CBU $\setminus\ \mathcal{P}^f_{g\ge 4}$), nor a super-case (as $\mathcal{P}_{g\ge 5}\setminus 2$-CBU is not empty, by Theorem~\ref{thm:W'_g}).

\section{Computational hardness for many problems}\label{sec:complexity}

We have seen (c.f. Theorem~\ref{thm:box-proper-CBU}, Corollary~\ref{cor:box-CBU}, and 
Corollary~\ref{cor:planar-subd}) that many 1-subdivided graphs belong to CBU, or even to 3- or 4-CBU.
For $(\ge 2)$-subdivided graphs, the picture is even simpler.
\begin{theorem}\label{thm:2-3-subd}
For every graph $G$, if we subdivide every edge at least twice, the obtained graph belongs to 3-CBU.
\end{theorem}
\begin{proof}
Let us denote $v_1,\ldots,v_n$ the vertices of $G$, and let $m=|E(G)|$. To construct a CBU 
representation for any $(\ge 2)$-subdivision, we start by assigning each vertex $v_i$ to the 
box $[3i,3i+1]\times [n-i,n-i+1]\times[0,2m]$. Then consider each edge $e$ of $G$ in any given 
order. For the $k^\text{th}$ edge $e$ assume it links $v_i$ and $v_j$, for some $i<j$, and assume 
$e$ is replaced by the path $(v_i,u_1,\ldots,u_r,v_j)$ for some $r\ge 2$. Here, $u_1$ is assigned to 
$[3i+1,3i+2]\times [n-j,n-i+1]\times[2k-1,2k]$, while the vertices $u_\ell$ with $2\le \ell \le r$
are assigned to $[3i+2+(\ell-2)(3j-3i-2)/(r-1),3i+2+(\ell-1)(3j-3i-2)/(r-1)]\times [n-j,n-j+1]\times[2k-1,2k]$ (see Figure~\ref{fig:2-3-subd}). One can easily check that the obtained representation is a 3-CBU representation of the subdivided graph.
\end{proof}
\begin{figure}
    \centering
    \includegraphics{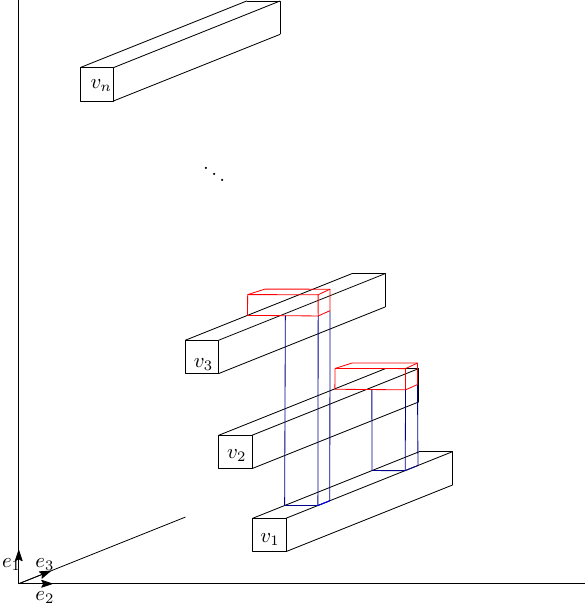}
    \caption{Construction of a 3-CBU representation of a $2$-subdivision of a graph.}
    \label{fig:2-3-subd}
\end{figure}

\begin{corollary}\label{cor:3-CBU-apx}
The problems of \textsc{Minimum Feedback Vertex Set} and 
\textsc{Cutwidth} are NP-hard, even when restricted to 3-CBU graphs.
The problems \textsc{Maximum Cut}, \textsc{Minimum Vertex Cover}, \textsc{Minimum Dominating Set}, and \textsc{Minimum Independent Dominating Set} are APX-hard, even when restricted to 3-CBU graphs.
\end{corollary}
\begin{proof}
For \textsc{Minimum Feedback Vertex Set} and \textsc{Cutwidth}, this follows from the fact that these problems are NP-hard, and that for any instance, subdividing an edge does not change the solution. For \textsc{Maximum Cut}, it follows from its APX-hardness and the fact that the maximum cut of a graph $G$ and its 2-subdivision $G_{2\text{-sub}}$ verify $mc(G) = mc(G_{2\text{-sub}})- 2|E(G)|$ and $3|E(G)|/2 = |E(G_{2\text{-sub}})|/2 \le mc(G_{2\text{-sub}}) \le |E(G_{2\text{-sub}})| = 3|E(G)|$. The other problems are shown APX-hard even when restricted to $6$-subdivided graphs~\cite{ChebikC07}.
\end{proof}

When restricted to 2-CBU some of these problems become simpler to handle, 
as every graph in 2-CBU is planar. 
Indeed, the \textsc{Maximum Cut} problem turns out to be polynomial time 
solvable~\cite{MCUT}, while \textsc{Minimum Vertex Cover}, 
\textsc{Minimum Dominating Set}, and \textsc{Minimum Independent Dominating Set}
admit PTAS~\cite{Baker,LiptonTarjan} (with standard techniques), such as 
\textsc{Minimum Feedback Vertex Set}~\cite{FVS}. 
However, many problems remain NP-hard when restricted to 2-CBU.

\begin{theorem}\label{thm:2-CBU-nph}
The problems \textsc{Maximum Independent Set}, \textsc{Minimum Vertex Cover}, \textsc{Minimum Dominating Set},  \textsc{Hamiltonian Path}, and \textsc{Hamiltonian cycle} are NP-complete,
even when restricted to 2-CBU graphs.
\end{theorem}
\begin{proof}
As these problems belong to NP, it remains to show that they are NP-hard for 2-CBU graphs.
Let us first show that the induced subgraphs of grids (so called \emph{grid graphs}) 
belong to 2-CBU. Consider the $n\times n$ grid $G$ such that
$V(G) = \{1,\ldots,n\}\times\{1,\ldots,n\}$, and such that the neighbors of any vertex 
$(i,j)$ are $\{(i,j-1)(i-1,j),(i,j+1),(i+1,j)\} \cap \{1,\ldots,n\}\times\{1,\ldots,n\}$.
Since it suffices to delete some boxes to obtain an induced subgraph,
the claim follows by constructing a 2-CBU representation for any such grid $G$. 
This construction is obtained by mapping any vertex $(i,j)$ to the box
$[i+j-1,i+j]\times[2i-2j,2i-2j+3]$ (see Figure~\ref{fig:grid2CBU}). 
As \textsc{Domination}~\cite{DOM}, \textsc{Hamiltonian Path}, and 
\textsc{Hamiltonian cycle}~\cite{HAM} are NP-hard for grid graphs, those problems 
are NP-hard for 2-CBU graphs.
\begin{figure}
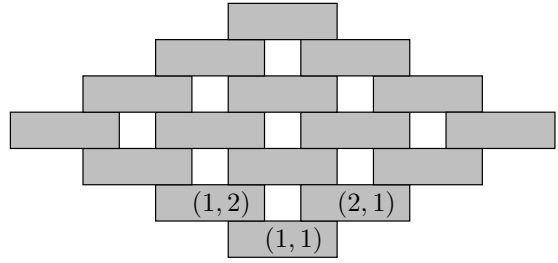

    \centering
   \ifx\XFigwidth\undefined\dimen1=0pt\else\dimen1\XFigwidth\fi
\divide\dimen1 by 6794
\ifx\XFigheight\undefined\dimen3=0pt\else\dimen3\XFigheight\fi
\divide\dimen3 by 3194
\ifdim\dimen1=0pt\ifdim\dimen3=0pt\dimen1=2000sp\dimen3\dimen1
  \else\dimen1\dimen3\fi\else\ifdim\dimen3=0pt\dimen3\dimen1\fi\fi
\tikzpicture[x=+\dimen1, y=+\dimen3]
{\ifx\XFigu\undefined\catcode`\@11
\def\temp{\alloc@1\dimen\dimendef\insc@unt}\temp\XFigu\catcode`\@12\fi}
\XFigu2000sp
\ifdim\XFigu<0pt\XFigu-\XFigu\fi
\clip(2228,-9472) rectangle (9022,-6278);
\tikzset{inner sep=+0pt, outer sep=+0pt}
\pgfsetlinewidth{+15\XFigu}
\pgfsetfillcolor{.!25}
\filldraw (4950,-9000) rectangle (6300,-9450);
\filldraw (5400,-9000) rectangle (4050,-8550);
\filldraw (4500,-8550) rectangle (3150,-8100);
\filldraw (4950,-8100) rectangle (6300,-8550);
\filldraw (5850,-8550) rectangle (7200,-9000);
\filldraw (6750,-8100) rectangle (8100,-8550);
\filldraw (7200,-8100) rectangle (5850,-7650);
\filldraw (5400,-8100) rectangle (4050,-7650);
\filldraw (3600,-8100) rectangle (2250,-7650);
\filldraw (3150,-7200) rectangle (4500,-7650);
\filldraw (4950,-7200) rectangle (6300,-7650);
\filldraw (6750,-7200) rectangle (8100,-7650);
\filldraw (4050,-6750) rectangle (5400,-7200);
\filldraw (4950,-6300) rectangle (6300,-6750);
\filldraw (5850,-6750) rectangle (7200,-7200);
\filldraw (7650,-7650) rectangle (9000,-8100);
\pgfsetfillcolor{.}
\pgftext[base,left,at=\pgfqpointxy{5400}{-9325}] {\fontsize{10}{21.6}\usefont{T1}{ptm}{m}{n}$(1,1)$}
\pgftext[base,left,at=\pgfqpointxy{6300}{-8875}] {\fontsize{10}{21.6}\usefont{T1}{ptm}{m}{n}$(2,1)$}
\pgftext[base,left,at=\pgfqpointxy{4500}{-8875}] {\fontsize{10}{21.6}\usefont{T1}{ptm}{m}{n}$(1,2)$}
\endtikzpicture%
        \caption{2-CBU representation of the $4\times 4$ grid.}
        \label{fig:grid2CBU}
\end{figure}

For the problems \textsc{Maximum Independent Set} and \textsc{Minimum Vertex Cover},
we have to consider a variant of grid graphs, the graph $R'(n_1,n_2)$ depicted 
in Figure~\ref{fig:R'2CBU}, and it is easy to see how to modify the construction above
in order to obtain a 2-CBU representation of this type of graphs. Again, this implies
that every induced subgraph of such a graph belongs to 2-CBU. As the problems 
\textsc{Maximum Independent Set} and \textsc{Minimum Vertex Cover} are NP-hard 
for this class (see the proof of Theorem 10 in~\cite{CLIQUE}), those problems are
NP-hard for 2-CBU graphs
\end{proof}
\begin{figure}
    \centering
    \includegraphics[scale=0.45]{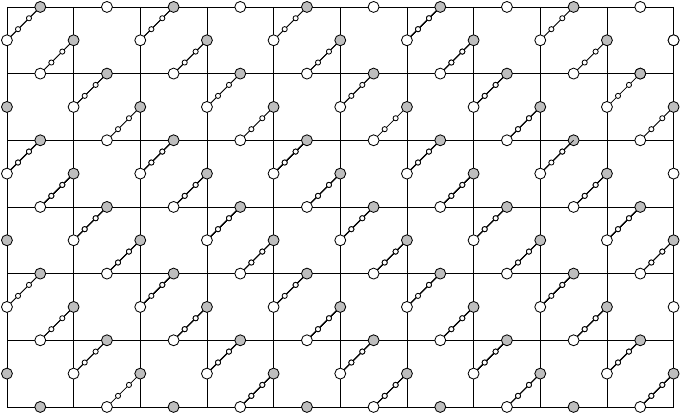}
\hspace{0.5cm}
\ifx\XFigwidth\undefined\dimen1=0pt\else\dimen1\XFigwidth\fi
\divide\dimen1 by 4994
\ifx\XFigheight\undefined\dimen3=0pt\else\dimen3\XFigheight\fi
\divide\dimen3 by 1394
\ifdim\dimen1=0pt\ifdim\dimen3=0pt\dimen1=2250sp\dimen3\dimen1
  \else\dimen1\dimen3\fi\else\ifdim\dimen3=0pt\dimen3\dimen1\fi\fi
\tikzpicture[x=+\dimen1, y=+\dimen3]
{\ifx\XFigu\undefined\catcode`\@11
\def\temp{\alloc@1\dimen\dimendef\insc@unt}\temp\XFigu\catcode`\@12\fi}
\XFigu2250sp
\ifdim\XFigu<0pt\XFigu-\XFigu\fi
\clip(3128,-9022) rectangle (8122,-7628);
\tikzset{inner sep=+0pt, outer sep=+0pt}
\pgfsetlinewidth{+15\XFigu}
\pgfsetfillcolor{.!25}
\filldraw (5400,-9000) rectangle (4050,-8550);
\filldraw (5850,-8550) rectangle (7200,-9000);
\filldraw (7200,-8100) rectangle (5850,-7650);
\filldraw (5400,-8100) rectangle (4050,-7650);
\filldraw (4950,-8100) rectangle (5400,-8325);
\filldraw (4950,-8325) rectangle (5400,-8550);
\filldraw (5850,-8100) rectangle (6300,-8325);
\filldraw (5850,-8325) rectangle (6300,-8550);
\filldraw (4500,-8100) rectangle (3150,-8550);
\filldraw (6750,-8100) rectangle (8100,-8550);
\pgfsetfillcolor{.}
\pgftext[base,left,at=\pgfqpointxy{4275}{-8815}] {\fontsize{9}{16.8}\usefont{T1}{ptm}{m}{n}$(i-1,j)$}
\pgftext[base,left,at=\pgfqpointxy{6075}{-8815}] {\fontsize{9}{16.8}\usefont{T1}{ptm}{m}{n}$(i,j-1)$}
\pgftext[base,left,at=\pgfqpointxy{6075}{-7915}] {\fontsize{9}{16.8}\usefont{T1}{ptm}{m}{n}$(i+1,j)$}
\pgftext[base,left,at=\pgfqpointxy{4275}{-7915}] {\fontsize{9}{16.8}\usefont{T1}{ptm}{m}{n}$(i,j+1)$}
\pgftext[base,left,at=\pgfqpointxy{3150}{-8365}] {\fontsize{8}{16.8}\usefont{T1}{ptm}{m}{n} $(i-1,j+1)$}
\pgftext[base,left,at=\pgfqpointxy{6750}{-8365}] {\fontsize{8}{16.8}\usefont{T1}{ptm}{m}{n} $(i+1,j-1)$}
\endtikzpicture%
    \caption{The graph $R'(n_1,n_2)$ and the local modification to obtain its 2-CBU representation. From the 2-CBU representation of the grid given above, one has to delete the box of every vertex $(i,j)$, where $i$ and $j$ are even, and if $i+j \equiv 2 \bmod 4$ one has to replace the box by 4 smaller boxes.}
    \label{fig:R'2CBU}
\end{figure}


\bibliographystyle{plain}
\bibliography{biblio.bib}

\end{document}